\documentclass[12pt,letterpaper]{article}

\usepackage{setspace}
\usepackage[
  margin=1.25in,
  papersize={8.5in,11in}, % US letter size
  headsep=0pt, % ensures header doesn't use vertical space
  footskip=0pt, % ensures footer doesn't use vertical space
  includefoot,
  includehead,
  verbose % verbose mode to see the final layout settings
]{geometry}
\usepackage{lipsum} % For generating placeholder text

\usepackage{amssymb,amsmath, amsthm, amsfonts,eurosym,geometry,ulem,graphicx,caption,color,setspace,sectsty,comment,footmisc,caption,natbib,pdflscape,subfigure,array,hyperref, bm}

\newtheorem{theorem}{Theorem}
\newtheorem{corollary}[theorem]{Corollary}
\newtheorem{proposition}{Proposition}
\newtheorem{lemma}{Lemma}
\newtheorem{claim}{Claim}
\newtheorem*{theorem*}{Theorem}
\newtheorem{definition}{Definition}

\newtheorem{assumption}{Assumption}

\newcolumntype{L}[1]{>{\raggedright\let\newline\\arraybackslash\hspace{0pt}}m{#1}}
\newcolumntype{C}[1]{>{\centering\let\newline\\arraybackslash\hspace{0pt}}m{#1}}
\newcolumntype{R}[1]{>{\raggedleft\let\newline\\arraybackslash\hspace{0pt}}m{#1}}

\begin{document}

\begin{titlepage}
\title{Testing Single Crossing Property with  Stochastic Choice Data\footnote{I thank Abhishek Arora, V Bhaskar, Svetlana Boyarchenko, Francesco Conti, Faheem Gilani, Raghav Malhotra, Maxwell Stinchcombe, Vasiliki Skreta, Caroline Thomas,  Thomas Wiseman, and Leeat Yariv for helpful comments. Laurens Cherchye, Bram De Rock, and Thomas Demuynck provided crucial help and guidance during my research visit to Universite Libré de Bruxelles, where part of this work was undertaken. Sydney Neal provided excellent research assistance.  Finally, I thank seminar and conference audiences at UT Austin, ULB, Delhi School of Economics, and Stony Brook University for helpful discussion and comments. I also acknowledge
the financial support received from 
Fonds de la Recherche Scientifique (FNRS), Belgium, which partly supported my research visit to ULB. Any remaining errors are my own.}}
\author{Tanay Raj Bhatt\footnote{trbhatt@utexas.edu, Department of Economics, University of Texas at Austin.}}
\date{September 2025}
\maketitle
\begin{abstract}
\noindent In a typical model of private information and choice under uncertainty, a decision maker observes a signal, updates her prior beliefs using Bayes rule, and maximizes her expected utility. If the decision maker's utility function satisfies the single crossing property, and the information structure is ordered according to the monotone likelihood ratio, then the comparative statics exhibit monotonicity with respect to signals. We consider the restrictions placed by this model of signal processing on state conditional stochastic choice data. In particular, we show that this model rationalizes a state conditional stochastic choice dataset if and only if the dataset itself is ordered according to the monotone likelihood ratio. A straightforward application of the main result shows the conditions under which the analyst can infer when one DM is more informed than the other. \\
\\
\vspace{0in}
\vspace{0in}
\bigskip
\end{abstract}
\setcounter{page}{0}
\thispagestyle{empty}
\end{titlepage}
\pagebreak \newpage

\doublespacing

In many models of choice under uncertainty, a decision maker (DM) is postulated to have preferences over action-state pairs, where the state is unobserved by the DM. The DM may have some private information, which is typically modeled as a mapping from the set of states to a set of signal realizations. If the DM is Bayesian, following a signal realization, she updates her prior beliefs using Bayes rule and the (known) structure of private information, and chooses an action to maximize her expected utility.  Such models of signal processing are now commonplace in economics. 

However, in an applied context, it is not always clear how such private information comes about. Indeed, a private signal may purely be an object of subjective considerations of the DM. If that is the case, the private signal would remain unobserved by an econometrician who wishes to perform some empirical exercise based on such a model. On the other hand, even if the private signal is ``objective" in some sense, the analyst may not have any information about how such a signal is acquired by the DM, and therefore, it may still remain unobserved.    

In an important contribution, \citealp{athey2002monotone} establishes conditions under which the comparative statics results in this setting exhibit monotonicity. In particular, if the DM's utility function satisfies the \textit{single crossing property} (SCP) and the information structure is ordered according to the \textit{maximum likelihood ratio} (MLR), then the DM's choice exhibits \textit{monotone comparative statics}\footnote{Essentially, single crossing property states that if higher action is preferred in a given state, then it must also be preferred at a higher state. If an information structure is MLR-ordered, then under higher signal realizations higher states are more  likely.} (MCS) - that is, higher signal realizations induce higher optimal actions.

The single crossing property is an ordinal generalization of the increasing differences property\footnote{A function $u$ exhibits increasing differences in $(x, \theta)$ if, $x'>x$ and $\theta'>\theta$ imply $u(x', \theta') - u(x, \theta') > u(x', \theta) - u(x, \theta)$. Clearly, this is a cardinal property. See \citealp{topkis1998supermodularity}.}. Since it implies monotone comparative statics, and due to its straightforward economic interpretation resting on strategic complementarity, the SCP has turned out to be an extremely important tool for applications - both in applied theory as well as in the empirical literature. For instance, in mechanism design, single crossing property allows us to pin down optimal transfers. In applied work, an intimately related property has been used to solve certain optimal taxation problems (see \citealp{mirrlees1971exploration}). In IO entry models, single crossing property is often invoked and facilitates  equilibrium analysis. In labor economics, the single crossing property is closely related to the so called Le Chatelier principle which helps characterize comparative static results (see \citealp{dekel2022comparative} for a recent work). The single crossing property also has some nice implications for preference aggregation, which proves to be quite useful in applied work (see, for instance \citealp{gans1996majority}).   

In this paper we ask, what restrictions does the single crossing property impose on certain kinds of  observable data? Often, available data - especially experimental data - may take a \emph{stochastic} form. The analyst may observe only some distribution(s) over  
available actions and may not possess any knowledge of either the DM's utility function or her private information. It is therefore important to understand whether the behavior of a Bayes rational agent imposes any restrictions on the stochastic data observed by the analyst. \citealp{caplin2015testable} show that rationalizability by a Bayes rational agent with access to some private information is equivalent to a single condition. Here, we ask what if the single crossing and MLR-ordered information structure impose any additional restrictions on top of those identified by \citealp{caplin2015testable}. The importance of this question is illustrated by the aforementioned examples. Since the SCP has proven to be such an important property, it seems worthwhile to see if it imposes meaningful testable restrictions on observable data.                  

Athey's results suggest that the stochastic data should exhibit certain kind of monotonicity. If $(i)$ the DM's utility function is single crossing; and $(ii)$ private information structure is MLR-ordered, then we know from \citealp{athey2002monotone} that $(iii)$  the comparative statics are monotone. How does $(i)-(iii)$ affect the data? In our setup, neither the utility function nor the private information structure is observed. Only a set of conditional distributions over actions are observed. If $(ii)$ holds, then higher states lead to higher signal realizations \textit{more often} - or with higher likelihood. If $(i)$ holds as well, then $(iii)$ is true, and so higher signals would lead directly to higher actions. One can expect this to be exhibited in data: since higher signals are more likely, and they lead to higher actions, we would expect \emph{higher actions to be more likely at higher states}. This monotonicity translates into the data itself being MLR-ordered - that is, higher actions have higher relative probabilities at higher states. Proposition \ref{thm:main_binary} shows that this intuition holds for the binary case. In fact, we establish something much stronger for the binary case. We show that this condition is not only \textit{both necessary and sufficient} for rationalization by an SCP utility function, but such data can \textit{only} be rationalized by a utility function satisfying the single crossing property. Furthermore, the main result of the paper shows that for any finite number of states, the MLR condition is both necessary and sufficient for rationalizability by a single crossing utility \textit{and} an MLR ordered information structure.   

The MLR condition is a rather strong restriction, and it should be stressed that what is actually being tested is not simply the hypothesis of a single crossing utility, but \emph{the joint hypothesis of a single crossing utility and an MLR-ordered information structure}. The main conjecture is therefore that rationalizability by a single crossing utility function \emph{and} MLR-ordered information structure is equivalent to a single condition, which requires the data to be MLR-ordered. Although the proof for the general case \textit{does} suggest a test for the single crossing property alone, we believe that a test for the joint hypothesis is more attractive in light of Athey's results.   

The main result suggests an interesting application. Because an MLR ordered dataset implies that that analyst may tentatively accept the hypothesis of a DM with single crossing utility and MLR ordered information. The Lehmann order ranks information structures based on their information content. In particular, one MLR-ordered information structure is Lehmann higher than the other if and only if, for every single crossing utility function, it yields a higher ex-ante expected utility. It is shown that, if two datasets are Lehmann ordered, then so are the information structures that rationalize them. This would allow the analyst to conclude that the DM to which  the former dataset corresponds is in some sense revealed to be more informed. That is, the result allows an analyst to perform inference on how informed DMs just from choice data.    

\textbf{Related literature.} Our paper tries to answer a question similar to ones asked within the literature on revealed preference theory, although the dataset we consider is quite different from finite data typically considered in classical revealed preference theory.   \citealp{green1991revealed}, \citealp{green1986expected}, \citealp{echenique2015savage}, \citealp{bayer2013ambiguity} are prominent examples of classical revealed preference theory for expected utility. However, this set of papers considers a much richer setting. They focus primarily on consumption or asset demand, while none explicitly feature information processing or consider the single crossing property. \citealp{caplin2015testable},  \citealp{doval2023core}, \citealp{rehbeck2023revealed} consider stochastic datasets and consider rationalizability by a Bayesian DM. The latter two consider rationalizability when the analyst only observes marginal distribution(s) over actions, but knows the DM's utility function. Our paper can thus be thought of as imposing particular shape restrictions in the setting of \citealp{caplin2015testable}.  \citealp{lazzati2018nonparametric} consider similar shape restrictions, but they prove a rationalizability theorem with finite data, where the variation in data comes from ``budgets" of actions. With finite data, they show that single crossing is indistinguishable from a weaker property developed by \citealp{quah2009comparative} - called interval dominance property. They also derive similar results for cross-sectional datasets and Bayesian games.

 \citealp{apesteguia2017single} posits a single crossing random utility model, dubbed SCRUM. The primary difference between that paper and ours is contextual: they deal with stochastic choice - there is a set of preferences with 
a probability measure defined on them that needs to be recovered. Here, the preferences are not stochastic - only the data is - and one is interested in the existence of \emph{some} preference and information structure consistent with the given data. Put another way, while the structure of the dataset available to the analyst is somewhat similar across these two groups of papers, the postulated \textit{data generating process} is different.

The next section presents the model, characterizes the dataset available to the analyst, and defines the corresponding notion of rationalizability. Section 2 presents the main result for the binary case - where the single crossing is indistinguishable from Bayesian expected utility maximization. Section 3 presents an example, and the two sections that follow establish necessity and sufficiency, respectively, for the general case. Section 6 shows how the main result can be used to perform inference on relative ``informedness" of DMs, and the final section presents concluding discussion.      

\section{The Model}

We begin this section with the canonical model of a Bayesian decision maker with access to some information structure. The implications of the single crossing property on monotone comparative statics are then presented. This is the model that the analyst wishes to test. The second subsection delineates the dataset that the analyst has access to, and the corresponding notions of rationalizability we are concerned with in this paper. 

\subsection{Preliminaries} 
Let $\Theta$ be a finite set of states and let $A$ be a finite set of actions, with $|A| \ge |\Theta|$. Typical elements of these sets are denoted by $\theta$ and $a$, respectively. Suppose that both $\Theta$ and $A$ are (strictly) ordered according to some exogenously given binary relations $>_\Theta, >_A$.\footnote{In particular, $>_\Theta, >_A$ are assumed to be strict linear orders. That is, they are complete, transitive, and contain no ``equalities". While the results on which this paper builds are given for partial orders, we restrict attention to linear orders for brevity. Presence of unordered alternatives/states/signals does not affect the results reported here.} To avoid notational clutter, in what follows, the subscripts are suppressed and the symbol $>$ is used to denote both the orderings. A decision maker (DM) wants to choose an action $a\in A$ without observing the state so as to maximize her expectation of her utility function $u:A\times \Theta\rightarrow \mathbb{R}$ over states.        
     
Let $S$ be some set of signal realizations, with $|S|\ge |\Theta|$. Let $S$ be ordered by some fixed binary relation $>_S$ as well. An information structure is a collection $\mathcal{I} \equiv \{\mu(\cdot|\theta)\}_{\theta\in \Theta}\subset \Delta (S)$. That is, each state induces a distribution over signal realizations. Let $\mu_0$ be the prior distribution over $\Theta$\footnote{To avoid unnecessary complications in the exposition, we assume that $\mu_0(\theta)>0$ for each $\theta\in\Theta$. This is without loss of generality since, if any state $\theta$ has a 0 probability, we can just drop that state and restrict attention to the set of states $\Theta/\{\theta\}$ instead. For revealed preference exercises, any such state is irrelevant.}. The decision maker is said to be Bayesian if she obtains a signal realization $s \in S$, forms a belief $\mu(\cdot|s)\in \Delta(\Theta)$ using Bayes rule, and chooses an action $a^*(s) \in \arg\max_A\mathbb{E}_{\mu(\cdot|s)}[u(a, \theta)]$. 

Each signal realization $s\in S$ induces a posterior distribution $\mu(\cdot|s) \in \Delta(\Theta)$. That is, if the signal realization observed by the DM is $s\in S$, then updating the prior belief using Bayes rule leads to the posterior distribution $\mu(\cdot|s)\in\Delta(\Theta)$.   Thus, ex-ante, the probability of distribution $\mu(\cdot|s)$ being induced is simply the probability of signal $s$ being observed, that is $\sum_{\theta\in \Theta}\mu_0(\theta)\mu(s|\theta)$. It is well known, at least since \citealp{blackwell1951comparison}, that there is a bijection between an information structure and the distribution over posteriors induced by it. Therefore, the information structure can equivalently be represented as a mapping $\pi:\Theta\rightarrow \Delta(\Delta(\Theta))$. Going forward, both these representations shall be employed interchangeably. When the latter representation is employed, $\pi(\gamma|\theta)$ will be used to denote the probability placed by the information structure on posterior belief $\gamma$ in state $\theta$.

One is often interested in conditions under which higher signal realizations translate to higher action choices by the DM - this refers to monotone comparative statics (MCS). That is, a Bayesian DM exhibits monotone comparative statics if,\footnote{However, the above definition of MCS is applicable only when there is always a unique optimal choice - which may not always be the case. In case of multiple solutions, sets of solutions are ordered using the \emph{strong set order} relation $\ge_{sso}$. In particular, given two sets $K_1, K_2$, say that $K_1 >_{sso} K_2$ if and only if for any $x' \in K_1, x\in K_2$, $\max\{x, x'\}\in K_1$ and $\min\{x, x'\}\in K_2$. Note that, when $K_1$ and $K_2$ are singletons, this relation reduces to $x'\ge x$.} 
\begin{align}
    s' > s \Rightarrow \arg\max_{a\in A}\mathbb{E}_{\mu(\cdot|s')}[u(a, \theta)] \ge \arg\max_{a\in A}\mathbb{E}_{\mu(\cdot|s)}[u(a, \theta)]. 
\end{align}

\citealp{athey2002monotone} provides conditions on $u$ and $\mathcal{I}$ under which monotone comparative statics (MCS) are obtained in this model. Some additional definitions are needed before Athey's result can be stated. 

\begin{definition} 
A utility function $u:A\times \Theta\rightarrow \mathbb{R}$ satisfies \emph{strict single crossing property (SCP)} in $(a, \theta)$ if,\footnote{The SCP is an ordinal generalization of increasing differences. However, while the property of increasing differences is symmetric, SCP is not. That is, if $u$ satisfies SCP in $(a, \theta)$, it need not, in general, be the case that $u(a', \theta'')> u(a', \theta')\Rightarrow u(a'', \theta'')> u(a'', \theta')$.} for states $\theta'' > \theta'$ and actions $a'' > a'$, 
\begin{align*}
    u(a'', \theta') \ge u(a', \theta') &\Rightarrow   u(a'', \theta'') \ge u(a', \theta''), \text{ and }\\
    u(a'', \theta') > u(a', \theta') &\Rightarrow   u(a'', \theta'') > u(a', \theta'').
\end{align*}
\end{definition}

The single crossing property says that if, for a given state a higher action is preferred, then the higher action must be preferred at higher states as well. 

Given two densities $\lambda$ and $\nu$ over some set $S\subset \mathbb{R}$, $\lambda$ is said to dominate $\nu$ in the \emph{MLR-order} if $\frac{\lambda(s)}{\nu(s)}$ is increasing in $s\in S$. Often, $\lambda$ is called an \emph{MLR shift} of $\nu$. The notion of MLR-shift is used to formalize the idea that \textit{higher states are more likely} under $\lambda$. It extends naturally to a family of conditional distributions:       

\begin{definition}
    A collection of conditional probability distributions $\{\nu(\cdot|s)\}\subset \Delta(\Theta)$ is said to be (strictly) \emph{MLR-ordered} if, for any $s, s'\in S, \theta, \theta'\in\Theta$ such that, $s' > s$ and $\theta' > \theta$, 
    \begin{align*}
            \frac{\nu(\theta'|s')}{\nu(\theta|s')} (>) \ge  \frac{\nu(\theta'|s)}{\nu(\theta|s)}
    \end{align*}
    An information structure $\mathcal{I}$ is said to be MLR-ordered if $\{\mu(\cdot|\theta)\}_{\theta\in\Theta}$ is MLR-ordered.
\end{definition}

So, if a collection of distributions over states conditional on signals is MLR-ordered, then, at higher signals, higher states are more likely. If the above inequality holds for all $\theta < \theta'$ and $s < s'$, write $\nu(\cdot|s') \succeq_{MLR} \nu(\cdot|s)$.\footnote{Note that, $\nu(\cdot|s') \succ_{MLR} \nu(\cdot|s)$ if and only if $\nu(\cdot|\theta') \succ_{MLR} \nu(\cdot|\theta)$, where $\nu(\cdot|\theta)$ is the distribution over signals conditional over states corresponding to $\nu(\cdot|s)$.}

We call a Bayesian DM with a single crossing utility function and an MLR-ordered information structure a \textit{Monontone Bayesian Expected Utility Decision Maker (MBEU-DM)}. Among other things, \citealp{athey2002monotone} shows that together, the two properties of SCP utility and MLR-ordered information imply monotone comparative statics. One version of the result is stated below.   

\begin{theorem*}{\textbf{(\citealp{athey2002monotone}, Theorem 2)}}
 Let $\Theta, A$ be partially ordered sets. Consider a DM characterized by $u$ with access to some MLR-ordered information structure $\mathcal{I}$. If $u$ satisfies the single crossing property, then the DM exhibits monotone comparative statics.\footnote{In fact, \citealp{athey2002monotone} shows something much stronger: $(i)$ MCS obtains for every single crossing utility if and only if the information structure is MLR-ordered and $(ii)$ MCS obtains for every MLR-ordered information structure if and only if the utility is single crossing. That is, together, the two conditions form a \emph{minimal pair of sufficient conditions} for MCS - neither of the two can be weakened any further. However, it must be noted that in the class of models considered in \citealp{athey2002monotone}, the set of available actions is not fixed, but changes with the parameter (in accordance with the strong set order). If the constraint set is taken to be fixed, \citealp{quah2009comparative} show that single crossing can be further weakened to a property they call \textit{interval dominance}.}  \label{thm:athey} 
\end{theorem*} 

Thus single crossing property imposes straightforward testable implications on observed choice - \textit{if the analyst observes the DM's private information}. The main contribution of the present paper is to show that, even if private information is unobserved, restrictions are nevertheless imposed on certain kind of datasets.

\subsection{Data and rationalizability} 

Without any knowledge of the utility function or the information structure available to the agent, the analyst wants to test for the single crossing property (and that the DM's private information is MLR-ordered).  The primary objective of this paper is thus to determine what restrictions are placed by single crossing utility on stochastic choice data, if any. This paper attempts to answer this question in the context of ``ideal" datasets, defined below.

\begin{definition}
    A \emph{dataset} is a collection of conditional distributions over actions - one for each state. That is, the available data takes the form $\bm{q} \equiv \{q(\cdot|\theta)\}_{\theta\in \Theta}\subset \Delta(A)$. 
\end{definition} 

In the terminology prevalent in revealed preference literature, this is \textit{somewhat} akin to the case of complete observability. The dataset is \textit{complete} in the sense that the analyst observes conditional distributions over actions for \textit{each} state. Therefore, the issue of unobserved counterfactuals is side-stepped. \footnote{However, this terminology has traditionally been only used to characterize deterministic datasets. In a deterministic environment, the integrability theorem is an instance of complete observability, while Afriat's theorem is an instance of partial observability. Cf footnote 3.}

In what follows, it will be assumed that the analyst knows the set of possible states $\Theta$, the set of actions, $A$, the prior $\mu_0$, and stochastic choice data $\bm{q}$. Thus, the tuple $\mathcal{D} \equiv \langle A, \Theta, \mu_0, \bm{q}\rangle$ denotes the \emph{observables} available to the analyst.\footnote{It is also being assumed that the (exogenously given) orders over the sets $A, \Theta$, and $S$ are known to the analyst.} This corresponds to the setting considered by \citealp{caplin2015testable}. There it is shown that the existence of a rationalizing utility function and an information structure is equivalent to a single condition on $\mathcal{D}$. Our paper investigates whether the single crossing property imposes further conditions on $\mathcal{D}$.

Note, however, that it is possible that the DM optimally mixes between different actions. To that end, rationalizing the observed data also requires a choice rule $C:\Delta(\Theta)\rightarrow \Delta(A)$, mapping posterior beliefs into \textit{distributions} over actions. Let $C(a|\gamma)$ denote the probability on action $a$ placed by the choice rule when the induced posterior is $\gamma$. Thus, an MBEU-DM is characterized by a tuple $\langle u, \mathcal{I}, C\rangle$, where $u$ is a strict single crossing utility function, $\mathcal{I}$ is an MLR-ordered information structure, and $C$ is a choice rule. 

We say that the data $\mathcal{D}$ is rationalized by an MBEU-DM - represented by $\langle u, \mathcal{I}, C\rangle$ - if the distribution over actions induced by the optimal behavior of the DM characterized by the tuple   \textit{is exactly} $\bm{q}$.  The following definition makes this notion of rationalizability precise,

\begin{definition}
\label{def:MBEUR}
    A collection of observables $\mathcal{D} \equiv \langle A, \Theta, \mu_0, \bm{q}\rangle$ is MBEU - rationalizable if there exists a tuple $\langle u, \mathcal{I}, C\rangle$ such that the following hold
    \begin{enumerate}
        \item Monotonicity: $u: A\times\Theta\rightarrow\mathbb{R}$ is a strict single crossing utility function and $\mathcal{I}$ is MLR-ordered.
        \item Bayes plausibility: The mapping $\pi:\theta\rightarrow \Delta(\Delta(\Theta))$, induced by the information structure $\mathcal{I}$, is Bayes plausible: 
        
        for every $\gamma\in \Delta(\Theta)$ such that $\pi(\gamma|\theta')>0$ for some $\theta'$,
        \begin{align*}
            \gamma(\theta) = \frac{\mu_0(\theta)\pi(\gamma|\theta)}{\sum_{\theta'}\mu_0(\theta')\pi(\gamma|\theta')} 
        \end{align*}
        \item Consistency: $q(a|\theta) = \displaystyle\sum_{\gamma\in \text{supp}(\pi)}\pi(\gamma|\theta)C(a|\gamma)$
        \item Optimality: For all $\gamma\in$ supp$(\pi)$, and $a\in A$ such that  $C(a|\gamma)>0$, 
        \begin{align*}
            \sum_{\theta\in \Theta}\gamma(\theta)u(a, \theta) \ge \sum_{\theta\in \Theta}\gamma(\theta)u(b, \theta), \hspace{5pt} \text{ for every $b\in A$}  
        \end{align*}
        where the inequality holds strictly for at least one action $b$.   
    \end{enumerate}
\end{definition}

The third condition demands that the observed data must in fact be induced by the choice behavior of the DM, while the fourth condition demands that the choice behavior be optimal. In certain degenerate cases with multiple optimal choices, it is possible that the choice rule \textit{does not} select an otherwise optimal action. This may result in misleading inference. To address this issue, we assume that each optimal action is taken with some positive probability.   

\begin{assumption} If $C$ is a choice rule corresponding to a DM with utility function $u$, then for each $\gamma\in \text{supp}(\pi), C(a|\gamma) > 0$ if and only if $a\in\arg\max\left\{\mathbb{E}_\gamma[u(a, \theta)]\right\}$      \label{as:as1}
\end{assumption}

While strong, the assumption is in line with the present context of complete observability. We now state the main result of the paper:

\begin{theorem}
Consider a dataset $\mathcal{D} = \langle A, \Theta, \mu_0, \bm{q}\rangle$. 
    \begin{enumerate}
        \item  If Assumption \ref{as:as1} holds and  $\mathcal{D}$ corresponds to an MBEU-DM, then ${q(\cdot|\theta)}_{\theta\in\Theta}$ is MLR-ordered.
        \item If ${q(\cdot|\theta)}_{\theta\in\Theta}$ is MLR-ordered, then $\mathcal{D}$ is MBEU-rationalizable.    
    \end{enumerate}
    \label{thm:Main_theorem} 
\end{theorem}

\section{Binary-state, binary-action case} 

In the two-states case, the single crossing property does not add any additional restrictions on top of those implied by  expected utility maximization. Let there be two states and two actions so that $\Theta = \{\theta_1, \theta_2\}$, $A = \{a_1, a_2\}$, $a_1<a_2$ and $\theta_1<\theta_2$, and let the prior be $\mu_0 = Pr(\theta = \theta_1)$. The utility function is characterized by four numbers $u(a_i, \theta_j), i, j = 1, 2$. The SCP is violated in this case if (and only if) $u(a_1, \theta_1) \le u(a_2, \theta_1)$ and $u(a_1, \theta_2) > u(a_2, \theta_2)$. Denote the points in the dataset $q(a_i|\theta_j)$ as $q^i_j$ - so that subscripts correspond to states and the superscripts correspond to actions. Similarly, denote the utility $u(a_i, \theta_j)$ by $u^i_j$. Note that any dataset in this case is characterized by two numbers: $q(a_1|\theta_1) \equiv q^1_1,  q(a_1|\theta_2) \equiv q^1_2$. From these, we can compute $q^2_1 = 1 - q^1_1$ and $q^2_2 = 1 - q_2^1$. The MBEU hypothesis then has a straightforward test as demonstrated in the following lemma.

\begin{proposition}
 Suppose that the data is characterized by $(q^1_1, q^1_2)$. The dataset is MBEU-rationalizable if and only if the data is MLR-ordered\footnote{In general, MLR property implies first order stochastic dominance. However, under the two statess case, MLR property is equivalent to the property of first order stochastic dominance. Moreover, in the above context, the dataset is MLR ordered if and only if $q^1_1 > q^1_2$.}.   \label{thm:main_binary}
\end{proposition} 

\begin{proof}
    See Appendix B.
\end{proof}

Thus, in the binary case, the MLR property alone is both necessary and sufficient for there to exist an MBEU-rationalization. The above result suggests a straightforward test of the SCP. Indeed, if the data is not MLR-ordered, then no BEU-rationalization of the data can satisfy the single crossing property. On the contrary, if the data is MLR-ordered, then there exists a BEU-rationalization that satisfies the SCP. The proof also exhibits something much stronger. It turns out that, in the binary case, an MLR-ordered dataset can be rationalized \textit{only} by an MBEU-DM. That is, if there is an information structure and a utility function that rationalizes the data, then the information structure \textit{must} be MLR-ordered and the utility function \textit{must} be single crossing. This result is stated below as a corollary.    

\begin{corollary}
    When there are only two states, and the data is MLR-ordered, data is rationalized only by single crossing utility functions.
\end{corollary}

The result says that, in 2 action, 2 state case, Bayesian expected utility hypothesis is indistinguishable from single crossing property. Put another way, the single crossing property puts no additional restricitions on the expected utility theory. These restrictions on  From an empirical perspective, this result is crucial, since many econometric application of single crossing property involve binary states (for instance, entry games). 
The next sections establish that Proposition 1 holds in the general case as well. 

\section{An Example}

Let there be three equi-likely states $\theta_1<\theta_2<\theta_3$ and three available actions $a_1<a_2<a_3$. Consider the utility function and an information structure with three signal realizations $l<m<h$ given in tables \ref{tab:tab_u} and \ref{tab:tab_I} respectively. 

\begin{minipage}{0.45\textwidth}
    \centering
    \captionof{table}{$u(a, \theta)$}\label{tab:tab_u}
        \begin{tabular}{c|c|c|c}
         & $a_1$ & $a_2$ & $a_3$\\
         \hline
         $\theta_1$ & 5 & -1 & 2\\
         $\theta_2$ & -1 & 5 & -1\\
         $\theta_3$ & 0 & -1 & 5\\
    \end{tabular}
\end{minipage}\hfill
\begin{minipage}{0.45\textwidth}
    \centering
    \captionof{table}{$\mathcal{I}\equiv \{\mu(s|\theta)\}$}\label{tab:tab_I}
        \begin{tabular}{c|c|c|c}
         & $l$ & $m$ & $h$\\
         \hline
         $\theta_1$ & 2/3 & 1/6 & 1/6\\
         $\theta_2$ & 1/4 & 1/2 & 1/4\\
         $\theta_3$ & 1/6 & 1/6 & 2/3\\
    \end{tabular}
\end{minipage}

The function $u$ exhibits two violations of the single crossing property in $(a, \theta)$: (1) $u(a_2, \theta_1) < u(a_3, \theta_1)$, $u(a_2, \theta_2) > u(a_3, \theta_2)$, and (2) $u(a_1, \theta_2) < u(a_2, \theta_2)$, $u(a_1, \theta_3) > u(a_2, \theta_3)$. 

Given the information structure, the distributions induced by the signals $l, m,$ and $h$ are $(8/13, 3/13, 2/13), (2/10, 6/10, 2/10)$, and $(2/13, 3/13, 8/13)$ respectively. These distributions are, in turn, induced with probabilities $13/36, 10/36, 13/36$ respectively. The optimal actions are $a^*(l) = a_1$, $a^*(m) = a_2$, $a^*(h) = a_3$. Thus, the induced state-conditional distributions over actions are, 
\begin{align}
    \mu(\cdot|\theta_1) &= (2/3, 1/6, 1/6)\nonumber\\
    \mu(\cdot|\theta_2) &= (1/4, 1/2, 1/4) \label{eq:ex_dist}\\
    \mu(\cdot|\theta_3) &= (1/6, 1/6, 2/3)\nonumber
\end{align}

In the present context, the stochastic data available to the analyst takes the form of equation \ref{eq:ex_dist}. The question this paper poses is, under what conditions on $\{\mu(\cdot|\theta)\}_{\theta_\in\Theta}$, does there exist an information structure, an SCP utility function, and a choice rule\footnote{Indeed, the choice rule may be degenerate, mapping posteriors into Dirac measures over actions.} consistent with the former in the sense of definition \ref{def:MBEUR}.        

Note, however, that the above utility function may be changed to a best-response equivalent utility function which has no violations of the SCP. One such utility function is depicted in table \ref{tab:tbl_SCPu} below.

\begin{table}
    \centering
    \captionof{table}{A $u(a,\theta)$ satisfying the SCP} \label{tab:tbl_SCPu}
        \begin{tabular}{c|c|c|c}
         & $l$ & $m$ & $h$\\
         \hline
         $\theta_1$ & 5  & 3 & 1\\
         $\theta_2$ & 3  & 5 & -1\\
         $\theta_3$ & -1 & 1 & 10\\
    \end{tabular}
\end{table}

\section{Proof of Necessity} 

For the rest of this section, let $u$ be a single crossing utility function, $\mathcal{I}$ be an MLR-ordered information structure, $\pi \in \Delta(\Delta(\Theta))$ be the distribution over posteriors induced by it, and $C:\text{supp}(\pi)\rightarrow \Delta(A)$ be a choice rule. Thus, the tuple $\langle u, \mathcal{I}, C\rangle$ corresponds to an MBEU-DM. 

Since $\mathcal{I}$ is MLR-ordered and $u$ is single crossing, MCS must hold. For any $\gamma\in\text{supp}(\pi)$, define the support of $C(\gamma)$ as $\text{supp}(C(\gamma)) := \{a \in A|C(a|\gamma) > 0\}$. Rationality requires that each action in the support of the choice rule must be optimal. Therefore, under Assumption \ref{as:as1}, the support must exactly be the solution set to the DM's maximization problem. MCS then requires that the solution must be ordered by the strong set order. The following establishes this argument formally.       

\begin{lemma}
If $\gamma, \gamma'$ are two posteriors induced by $\mathcal{I}$ such that $\gamma' \succ_{MLR} \gamma$, then $\text{supp}(C(\gamma'))$ $\ge_{sso} \text{supp}(C(\gamma))$.     \label{lm:nec_lemma} 
\end{lemma}
\begin{proof}
    Since $u$ is single crossing and $\mathcal{I}$ is MLR-ordered, MCS holds by Theorem \ref{thm:athey}. Since $\gamma', \gamma$ are induced by $\mathcal{I}$, if  $\gamma'\succ_{MLR}\gamma$ it must be the case that $s'>s$, where $s'$ and $s$ are signal realizations corresponding to $\gamma'$ and $\gamma$ respectively. But by MCS, this implies $\arg\max\mathbb{E}_{\gamma'}[u] \ge_{sso} \arg\max\mathbb{E}_{\gamma}[u]$. Since $C$ is a choice rule, for any belief $\lambda$ and action $a\in A$, $C(a|\lambda) > 0$ only if $a \in \arg\max\mathbb{E}_{\lambda}[u]$. Hence, it must be that $\text{supp}(C(\gamma')) \ge_{sso} \text{supp}(C(\gamma))$. \\
\end{proof}

Suppose, for now, that the choice rule is degenerate at each belief. That is, for each $\gamma\in\text{supp}(\pi)$, $C(\gamma)=\delta_a$ for some $a$. By Lemma \ref{lm:nec_lemma}, $C$ is increasing in the sense that if $\gamma'\succ_{MLR} \gamma$, $C(\gamma') = \delta_{a'}$, and $C(\gamma) = \delta_a$, then $a' \ge a$. For each $a\in A$, define $\Gamma(a) := \{\gamma\in\text{supp}(\pi)|C(\gamma) = \delta_a\}$. That is, $\Gamma(a)$ is the set of all posteriors at which $a$ is optimal. Then, the probability of observing action $a$ in state $\theta$ is given by $\tilde{q}(a|\theta) := \sum_{\gamma\in\text{supp}(\pi)}\pi(\gamma|\theta)C(a|\gamma) = \sum_{\gamma\in\Gamma(a)}\pi(\gamma|\theta)$. Obviously, data matching requires that $\tilde{q}(a|\theta)$ coincides with the observed data. Finally, let $\hat{A} := \{a\in A|\exists \gamma\in \text{supp}(\pi) \text{ such that } C(\gamma) = \delta_a\}$ - so that $\hat{A}$ is the set of actions which are optimal for some posterior. Put another way, $\hat{A}$ is the set of actions that may be (optimally) chosen by the DM.      

Lemmata \ref{lm:nec_det} and \ref{lm:nec_sto} establish the consequences of monotone comparative statics on observed data for the general case. MCS says that, given a single crossing utility function and an MLR-ordered information structure, higher signals induce rational agents to take higher actions. Since the data is MLR-ordered, in higher states, higher signal realizations are \textit{more likely} or, have higher odds. Because higher signals in turn induce higher actions, higher actions must be observed at higher states with higher odds as well. But this simply means that the collection $\{\tilde{q}(\cdot|\theta)\}_{\theta\in \Theta}$ is itself MLR-ordered. 

For the case where the choice rule is deterministic - so that for each posterior there is a unique optimal action - Lemma \ref{lm:nec_det} establishes this argument formally. In this case, the result rests upon the observation that the posteriors which induce a higher action must be higher - in the the sense of MLR - than those that induce a lower action. This follows from the contrapositive of MCS and our assumption that the sets of signals, states, and actions are \textit{strictly linearly ordered}. Suppose $\theta_2>\theta_1$ and $a_2>a_1$. Then, the MLR property of $\mathcal{I}$ implies that, under $\theta_2$, the likelihood of obtaining a posterior that induces $a_2$ relative to one that induces $a_1$ is higher than under $\theta_1$. Moreover, this holds for all such pairs of states and actions, which essentially means that the collection of conditional distributions over actions is MLR-ordered.  

\begin{lemma}
    If $C(\gamma)$ is degenerate at every $\gamma$, $u$ satisfies SCP, and $\mathcal{I}$ is MLR-ordered, then, $\{\tilde{q}(\cdot|\theta)\}_{\theta\in \Theta}$ is MLR-ordered.   \label{lm:nec_det}
\end{lemma}
\begin{proof}
Again, MCS holds. Let $a', a\in\hat{A}$ such that $a' > a$. Then, for a given $\theta$,  
    \begin{align*}
        \frac{\tilde{q}(a'|\theta)}{\tilde{q}(a|\theta)} = \frac{\sum_{\gamma\in\Gamma(a')}\pi(\gamma|\theta)}{\sum_{\gamma\in\Gamma(a)}\pi(\gamma|\theta)} = \frac{\sum_{\gamma\in\Gamma(a')}\mu(s_\gamma|\theta)}{\sum_{\gamma\in\Gamma(a)}\mu(s_\gamma|\theta)} 
    \end{align*}
    where $s_\gamma$ denotes the signal realization that induces the signal $\gamma$.\footnote{Note that this is without loss of generality. If, in a given information structure, multiple signal realizations induce same posterior, then they can simply be ``collapsed" together.} By MCS, $a < a'$ implies $s_{\gamma} \le s_{\gamma'}$ for any $\gamma'\in \Gamma(a')$ and $\gamma \in \Gamma(a)$. As $C(\gamma)$ is degenerate for each $\gamma$, $\Gamma(a)\cap\Gamma(a') = \emptyset$. Moreover, as argued above, it is without loss of generality to assume that each signal realization induces a unique posterior, and therefore it must be that $s_\gamma < s_{\gamma'}$. The latter follows from the assumption that the set of signals is completely ordered by a strict linear order. Now, if $\theta'>\theta$, for $\{\tilde{q}(\cdot|\theta)\}_{\theta\in\Theta}$ to be MLR-ordered one needs
    \begin{align*}
        \tilde{q}(a'|\theta')\tilde{q}(a|\theta) &\ge \tilde{q}(a'|\theta)\tilde{q}(a|\theta')\\
        \Leftrightarrow  \left(\sum_{\gamma\in\Gamma(a')}\mu(s_\gamma|\theta')\right)\left(\sum_{\gamma\in\Gamma(a)}\mu(s_\gamma|\theta)\right) &\ge \left(\sum_{\gamma\in\Gamma(a')}\mu(s_\gamma|\theta)\right)\left(\sum_{\gamma\in\Gamma(a)}\mu(s_\gamma|\theta')\right)\\
        \Leftrightarrow \sum_{(\gamma, \gamma')\in\Gamma(a)\times\Gamma(a')}\mu(s_{\gamma'}|\theta')\mu(s_\gamma|\theta) &\ge \sum_{(\gamma, \gamma')\in\Gamma(a)\times\Gamma(a')}\mu(s_{\gamma'}|\theta)\mu(s_\gamma|\theta').
    \end{align*}
    Since $\mu(s|\theta)$ is MLR-ordered, any individual term from left side of the last inequality must be at least as large as any on the right. Thus, the only case of concern is one where the number of terms on RHS is greater than those on the LHS. This happens when there are pairs of posteriors $(\gamma, \gamma')\in\Gamma(a)\times\Gamma(a')$ such that the $\mu(s_{\gamma'}|\theta)\mu(s_\gamma|\theta')>0$ while $\mu(s_{\gamma'}|\theta')\mu(s_\gamma|\theta) = 0$. However, if $\mu(s_{\gamma'}|\theta') = 0$ then $\mu(s_{\gamma'}|\theta)$ must be $0$ by  MLR. Similarly, if $\mu(s_{\gamma}|\theta) = 0$ then $\mu(s_{\gamma}|\theta')$ must be $0$. Therefore, the fact that $\{\mu(\cdot|\theta)\}_{\theta\in \Theta}$ is MLR-ordered ensures that the above inequality holds and thus $\{\tilde{q}(\cdot|\theta)\}_{\theta\in \Theta}$ is MLR-ordered.      \\
\end{proof}

Following is the main result of this section. It builds on Lemma \ref{lm:nec_det} and shows that the argument goes through even when the choice rule is non-degenerate - that is, when there are posteriors at which there are multiple optimal solutions. The argument uses \ref{lm:nec_det} and the fact that the support of the choice rule is ordered by the strong set order. Therefore, actions induced by a higher (in the sense of MLR) posterior must dominate - in the strong set order - the actions induced by a lower posterior. Following this observation, the role played by choice rule turns out to be redundant, and the MLR inequality is established for the general case.         

\begin{lemma}
    If $\langle u, \mathcal{I}, C\rangle$ correspond to an MBEU-DM, then the collection of conditional distributions over actions, $\{\tilde{q}(\cdot|\theta)\}_{\theta\in \Theta}$, is MLR-ordered.  \label{lm:nec_sto}
\end{lemma}

\begin{proof}
    First note that, if $\gamma' \succ_{MLR}\gamma$, it must be that $s_{\gamma'} > s_\gamma$. For any action $a$, define $\tilde{\Gamma}(a) = \{\gamma\in \text{supp}(\pi)|C(a|\gamma) > 0\}$. Consequently, if $a' > a$ then $\{s_{\gamma'}|\gamma'\in\tilde{\Gamma}(a')\} \ge_{sso} \{s_{\gamma}|\gamma\in\tilde{\Gamma}(a)\}_{\theta\in\Theta}$. For $\{q(\cdot|\theta)\}$ to be MLR-ordered, one needs,     
    \begin{align*} \left(\sum_{\gamma\in\text{supp}(\pi)}\mu(s_\gamma|\theta')C(a'|\gamma)\right)&\left(\sum_{\gamma\in\text{supp}(\pi)}\mu(s_\gamma|\theta)C(a|\gamma)\right) \\> \left(\sum_{\gamma\in\text{supp}(\pi)}\mu(s_\gamma|\theta)C(a'|\gamma)\right)&\left(\sum_{\gamma\in\text{supp}(\pi)}\mu(s_\gamma|\theta')C(a|\gamma)\right)
    \end{align*}
    \begin{align*}
    \Leftrightarrow \left(\sum_{\gamma'\in\tilde{\Gamma}(a') }\mu(s_{\gamma'}|\theta')C(a'|\gamma')\right)\left(\sum_{\gamma\in\tilde{\Gamma}(a) }\mu(s_\gamma|\theta)C(a|\gamma)\right) &\\> \left(\sum_{\gamma'\in\tilde{\Gamma}(a') }\mu(s_{\gamma'}|\theta)C(a'|\gamma')\right)&\left(\sum_{\gamma\in\tilde{\Gamma}(a) }\mu(s_\gamma|\theta')C(a|\gamma)\right)\\
   \Leftrightarrow \sum_{(\gamma, \gamma')\in\tilde{\Gamma}(a)\times\tilde{\Gamma}(a')}\mu(s_{\gamma'}|\theta')\mu(s_\gamma|\theta)C(a'|\gamma')C(a|\gamma) &> \sum_{(\gamma, \gamma')\in\tilde{\Gamma}(a)\times\tilde{\Gamma}(a')}\mu(s_{\gamma'}|\theta)\mu(s_\gamma|\theta')C(a'|\gamma')C(a|\gamma)
    \end{align*}
    In the last inequality, terms corresponding to the case where $\gamma = \gamma'$ can be eliminated from both sides. The only remaining terms will thus be the ones where $s_{\gamma'} > s_{\gamma}$. For each of these pairs, the terms $\mu(s_{\gamma'}|\theta')\mu(s_\gamma|\theta)$ are greater than the terms $\mu(s_{\gamma'}|\theta)\mu(s_\gamma|\theta')$ by MLR property and Lemma \ref{lm:nec_det}, and the result follows. 
\end{proof}

As pointed out in the introduction, the necessary conditions are intuitive in light of Athey's results. Hence, the more interesting question is whether they are also sufficient. The next section answers this question in the affirmative. Given an MLR-ordered dataset, constructing an MLR-ordered information structure is straightforward. Establishing the existence of a single crossing utility turns out to be more delicate.

\section{Proof of Sufficiency} 

We now show that if the data is MLR ordered, then there exists a single crossing utility function $u$, an MLR ordered information structure $\pi$, and a choice rule $C$ that rationalize the data. The next subsection shows the construction of the choice rule and information structure.

\subsection{Constructing a single crossing utility function}
  
Let there be $M$ actions and $N$ states, with $M\ge N$. Enumerate the actions and states $A = \{a_1, \dots, a_M\}, \Theta = \{\theta_1, \dots, \theta_N\}$, so that $a_{l'}>a_l, \theta_{l'}>\theta_l$ whenever $l'>l$. For any action $a$, define $\overline{\Theta}(a)$ as the set of states with the highest  conditional probability of $a$,  
\begin{align*}
    \overline{\Theta}(a) = \arg\max_{\theta}\{q(a|\theta)\}
\end{align*}
So, for any $\theta\in \overline{\Theta}(a), q(a|\theta)\ge q(a|\theta')$ for any $\theta'\in \Theta$. 

\begin{lemma}
    Suppose $\{q(\cdot|\theta)\}_{\theta\in\Theta}$ is MLR ordered. Then, $a'>a$ implies $\overline{\Theta}(a')>_{sso} \overline{\Theta}(a)$.  \label{thm:sso_lemma}
\end{lemma}

\begin{proof}
Suppose not, so that for some $a', a$ with $a'>a$, there exist $\overline{\theta}(a')\in \overline{\Theta}(a'),   \overline{\theta}(a)\in \overline{\Theta}(a)$, such that $\overline{\theta}(a')< \overline{\theta}(a)$, and either $\overline{\theta}(a')\notin \overline{\Theta}(a)$, or $\overline{\theta}(a)\notin \overline{\Theta}(a')$. 
    
Suppose $\overline{\theta}(a')\notin \overline{\Theta}(a)$. The other case is handled in a similar manner. By definition, $q(a|\overline{\theta}(a)) > q(a|\overline{\theta}(a'))$ and $q(a'|\overline{\theta}(a')) \ge q(a'|\overline{\theta}(a))$. It follows that, 
\begin{align*}
    \frac{q(a|\overline{\theta}(a))}{q(a|\overline{\theta}(a'))} > \frac{q(a'|\overline{\theta}(a))}{q(a'|\overline{\theta}(a'))}
\end{align*}
Since, by assumption, $a'>a$ and $\overline{\theta}(a')< \overline{\theta}(a)$, the above inequality violates MLR.   
\end{proof}

Now, consider a finite, (weakly) increasing integer sequence $\{n_{a_k}\}_{k=1}^M$,  avoiding ties wherever possible, and such that $\theta_{n_{a_k}} \in \overline{\Theta}(a_k), k=1, \dots, M$. Doing this is possible in light of lemma \ref{thm:sso_lemma}.\footnote{One way to do this, for instance, is to define $n_{a_1} = \min\{j|\theta_j\in\overline{\Theta}(a_1)\}$, and then inductively define $n_{a_k} = \min\{j|\theta_j\in\overline{\Theta}(a_k)\backslash\theta_{n_{a_{k-1}}}\}$, if $\overline{\Theta}(a_k)\backslash\overline{\theta}_{n_{a_{k-1}}}\neq \emptyset$; otherwise, set $n_{a_k} = n_{a_{k-1}}$. The sequence is weakly increasing since the sets $\overline{\Theta}(a_k)$ is ordered in SSO.}.     

In order to prove the result, we rely on Theorem 1 of \citealp{caplin2015testable}. They show that rationalizability by a Bayesian expected utility maximizer requires the existence of a solution to a certain system of inequalities. The solution to this system is in turn a rationalizing utility function. Since the objective here is to prove the existence of a rationalizing utility that is also single crossing, it is sufficient to show that, when the dataset is MLR ordered, there exists a single crossing utility that solves the aforementioned system of inequalities.        

We proceed to construct a SC utility function inductively. For each consecutive pair $a+1, a$, we compute single crossing differences $\{A^\theta_{a+1, a}\}_{\theta\in\Theta}$ and then find a utility function with those differences. However, an additional difficulty that needs to be addressed when $N>3$ concerns aggregation. For instance, even if $\{A^\theta_{a+1, a}\}$ and $\{A^\theta_{a, a-1}\}$ are single crossing, $\{A^\theta_{a+1, a-1}\} = \{A^\theta_{a+1, a}+A^\theta_{a, a-1}\}$ may not be so. That is, the single crossing property is not necessarily preserved under aggregation. However, \citealp{quah2012aggregating} report the condition - called the \textit{signed ratio property} - under which the single crossing property \textit{is} preserved under aggregation. These can be written as linear restrictions and  incorporated into the systems of equations that characterize the single crossing differences. 

Indeed, these issues arise only when there are more than three states, and the number of restrictions that need to be imposed is increasing in the number of states. Before we proceed to inductively construct a single crossing utility function, we note that, for any pair of consecutive actions $a, a+1$, we can pick the state at which the function $\theta\mapsto A^\theta_{a+1, a}$ crosses the horizontal axis. In this case, we construct the differences $A^\theta_{a+1, a}$ so that the switch happens at $\theta_{n_{a+1}}$.    

Denote by $\bm{q(a_i|\theta)}$ the row vector $ (q(a_i|\theta_1), \dots, q(a_i|\theta_N))$, and $\bm{1}_{\theta'}$ the unit (row) vector with a 1 in the position $j$, where $j$ is such that $\theta_j = \theta'$, and 0 elsewhere. We proceed inductively. For $k = 2$, consider the following system:  

\begin{align}
    \begin{bmatrix}
        \bm{q(a_k|\theta)} \\
        -\bm{q(a_{k-1}|\theta)}\\
        -\bm{1}_{\theta_1} \\
         \vdots \\
        -\bm{1}_{\theta_{n_{a_k} - 1}}\\
         \bm{1}_{\theta_{n_{a_k}}}\\
         \vdots\\
         \bm{1}_{\theta_N}
    \end{bmatrix}_{(N+2)\times N}\cdot\begin{bmatrix}
        A^1_{k, k-1} \\ A^2_{k, k-1} \\\vdots\\ A^N_{k, k-1} 
    \end{bmatrix} \ge 0. \label{eq:sys_SC1}
\end{align}

It follows from the MLR property and a Theorem of Alternative that this system has a solution. In particular, consider the following version of the Theorem of Alternative.

\begin{theorem}
Exactly one of the following systems has a solution\footnote{There are many 
 different versions, and proofs of the Theorem of Alternative, sometimes also known as Farkas' Lemma. See for instance \citealp{gale1989theory}.}:
    \begin{align*}
    A'x = \textbf{0} \hspace{20pt} x \gg \textbf{0} \hspace{40pt} &\textbf{(I)}\\
    Ay \ge \textbf{0}, y\neq 0 \hspace{60pt} &\textbf{(II)}. 
\end{align*}\label{thm:ThmofAlt}
\end{theorem}

Then, we have the following.

\begin{claim}
    The system in Eq \ref{eq:sys_SC1} has a solution.  
\end{claim}

\begin{proof}
   Consider the system, 
   \begin{align*}
       A'x = 0, x\gg 0
   \end{align*}
where $A$ is the first matrix on the left hand side of \ref{eq:sys_SC1}. Suppose that this system does have a solution $\bm{x} = (x_1, \dots, x_{N+2}) \gg 0$. Then, for any $\theta = \theta_j< \theta_{n_{a_k}}, q^\theta_kx_1 - q^\theta_{k-1}x_2 - x_j = 0$. Since $x_j > 0, q^\theta_kx_1 - q^\theta_{k-1}x_2 > 0$. That is, $\frac{q^\theta_k}{q^\theta_{k-1}} > \frac{x_2}{x_1}$. Then, for any $\theta' = \theta_l \ge \theta_{n_{a_k}} > \theta$, $\frac{q^{\theta'}_k}{q^{\theta'}_{k-1}} > \frac{q^\theta_k}{q^\theta_{k-1}} > \frac{x_2}{x_1}$, where the first inequality follows from the MLR property. Therefore, $q^{\theta'}_kx_1-q^{\theta'}_{k-1}x_2>0$. But $\bm{x}$ is a solution to the dual system, and so $q^{\theta'}_kx_1-q^{\theta'}_{k-1}x_2 + x_l = 0$ - which is a contradiction since $x_l > 0$. 

Thus, the system $A'x = 0$ does not have a solution $x\gg 0$. By theorem \ref{thm:ThmofAlt}, the system \ref{eq:sys_SC1} must have a solution.    
\end{proof}

So, there exist solutions $\{A^\theta_{2, 1}\}_{\theta=1}^N$. Note that because of the last $N-2$ constraints in \ref{eq:sys_SC1}, these two solutions satisfy single crossing by construction.

Now, for each $k=3, 4, \dots$, suppose we have $\{A^\theta_{k-1, l}\}_{\theta\in\Theta}, l=1, \dots, k-2$, each of which satisfy the single crossing property. We follow the steps given below: 
\begin{enumerate}
    \item Define $\overline{\Theta}(a_k)$ as before. For each $l=1, \dots, k-1$, define $\widetilde{\Theta}^{k, l} := \{\theta\in\Theta|\theta<\theta_{n_{a_k}}, A^\theta_{k-1, l} > 0\}$. 
    \item For each $l=1, \dots, k-2$, define $\widetilde{A}^{k} = \{a_l\in A||\widetilde{\Theta}^{k, l}| \ge 2\}$ (these are the only states for the pair $k, l$, for which the aggregation issue is non-trivial).  
    \item For each $l$ such that $a_l\in \widetilde{A}^k$, and for each $\theta', \theta''\in\widetilde{\Theta}^{k, l}$ such that $\theta''>\theta'$, define the row vector
    \begin{align}
        \bm{\kappa^l(\theta', \theta'')} := \bm{1}_{\theta''}A^{\theta'}_{k-1, l} - \bm{1}_{\theta'}A^{\theta''}_{k-1, l}  \label{eq:agg_row_vec}
    \end{align}
    (this is essentially the signed ratio property of Quah and Strulovici which preserves SCP under aggregation).
    \item Let $\bm{\kappa(\widetilde{\Theta}^{k, l})}$ be the matrix obtained by stacking all the row vectors \ref{eq:agg_row_vec} on top of each other. 
    \item Let $\bm{\kappa(\widetilde{A}^k)}$ be the matrix obtained by stacking matrices $\{\bm{\kappa(\widetilde{\Theta}^{k, l})}\}_{l|a_l\in \widetilde{A}^k}$ on top of each other.  
    \item Now, consider the system, 
    \begin{align*}
\begin{bmatrix}
    \bm{q^\theta_k}\\
    -\bm{q^\theta_{k-1}}\\
    -\bm{1}_{\theta_1}\\
    \vdots \\
    -\bm{1}_{\theta_{n_{a_k}-1}}\\
    \bm{1}_{\theta_{n_{a_k}}}\\
    \vdots\\
    \bm{1}_{\theta_N}\\
    \bm{\kappa(\widetilde{A}_k)}
\end{bmatrix}
\cdot
    \begin{bmatrix}
        A^1_{k, k-1}\\
        A^2_{k, k-1}\\
        \vdots\\
        A^N_{k, k-1} 
    \end{bmatrix}\ge 0
\end{align*}
By virtually the same arguments as above, this system has a solution by MLR property and the Theorem of Alternative. 
\item Finally, for each $l = 1, \dots, k-2$, define $A^\theta_{k, l} = A^\theta_{k, k-1} + A^\theta_{k-1, l}$.
\end{enumerate}

All the possible differences $\{A^\theta_{k, l}\}_{\theta\in\Theta}, k=2, \dots, M, l = 1, \dots, k-1$ are thus obtained. However, note that the solutions $\{A^\theta_{k, l}\}$ may only be \textit{weakly single crossing}.\footnote{That is, if  $\theta'>\theta$ then $A^\theta_{k, k-1}\ge 0$ implies $A^{\theta'}_{k, k-1}\ge 0$ only.} However, once any such collection of weakly single crossing differences are obtained, it is straightforward to convert them into a collection of single crossing differences. So, it is without loss of generality to assume that the solutions $\{A^\theta_{k, l}\}$ are, in fact, single crossing.  

We now show that $\{A^\theta_{k, l}\}_{\theta\in\Theta}$
 are indeed appropriate SC differences. We begin with an important observation: 

\begin{claim}
    Consider any $\{A^\theta_{k-1, l}\}_{\theta\in\Theta}, l = 1, \dots, k-2$. Then, for all $\theta \ge \theta_{n_{a_{k-1}}}, A^\theta_{k-1, l} \ge 0$.\label{thm:agg_req}  
\end{claim}  
\begin{proof}
    Note that,
    \begin{align*}
        A^\theta_{k-1, l} &= A^\theta_{k-1, k-2} + A^\theta_{k-2, l}\\
        &= A^\theta_{k-1, k-2} + A^\theta_{k-2, k-3} + \dots + A^\theta_{2, 1}.  
    \end{align*}
By construction, for each $A^\theta_{a_i, a_i-1}$, the switch happens at $\theta_{n_{a_i}}$. Since $\theta_{n_{a_i}}$ is a (weakly) increasing sequence, for $\theta \ge \theta_{n_{a_{k-1}}}$, each of the differences in the above expression are nonnegative.   
\end{proof}

\begin{claim}
    For any $l = 1, \dots, k-1, \{A^\theta_{k, l}\}_{\theta\in\Theta}$ satisfies SCP.  
\end{claim}

\begin{proof}
    For $l = k-1, k-2$, SCP follows by construction. For $l < k-2$, and by Claim \ref{thm:agg_req}, we only need to take care of cases with $A^\theta_{k, k-1} < 0, A^\theta_{k-1, l} > 0$. These are exactly the states $\theta$ that lie in $\widetilde{\Theta}^{k, l}$ and the restrictions involving $\bm{\kappa_{\widetilde{A}_k}}$ ensure that the appropriate version of signed ratio property holds for all such action-state pairs. SCP then follows by Proposition 1 of \citealp{quah2012aggregating}.       
\end{proof}

We now show that the differences are ``optimal". The optimality conditions are precisely the NIAS inequalities from \citealp{caplin2015testable}, which can be rewritten in terms of utility differences.  

\begin{claim}
    For each $k$, and for each $l = 1, \dots, k-1$, $\{A^\theta_{kl}\}_{\theta\in\Theta}$ satisfy optimality. That is, for each $k=1, \dots, M$, and $l=1, \dots, k-1$, 
    \begin{align}
        \sum_{\theta}q^\theta_{k}A^\theta_{k, l} \ge 0\label{eq:opt_high}\\
        \sum_{\theta}q^\theta_{l}A^\theta_{k, l} \le 0.\label{eq:opt_low}
    \end{align}
\end{claim}

\begin{proof}
For $l = k-1$, optimality follows by construction. Now consider any $l<k-1$. Recall that $A^\theta_{k, l} = A^\theta_{k, k-1} + A^\theta_{k-1, k-2} + \dots + A^\theta_{l+1, l}$. We use induction on $l$ to prove the result. First, consider $A^\theta_{k, k-2} = A^\theta_{k, k-1} + A^\theta_{k-1, k-2}$. By construction, $\sum_{\theta}q^\theta_{k}A^\theta_{k, k-1} \ge 0, \sum_{\theta}q^\theta_{k-1}A^\theta_{k, k-1} \le 0, \sum_{\theta}q^\theta_{k-1}A^\theta_{k-1, k-2} \ge 0$ and, $\sum_{\theta}q^\theta_{k-2}A^\theta_{k-1, k-2} \le 0$. Now consider,   
\begin{align}
\left(q^1_{k-1}\frac{q^{\theta_{n_{a_{k}}}}_{k-2}}{q^{\theta_{n_{a_{k}}}}_{k-1}}\right)A^1_{k, k-1} + \dots + \left(q^N_{k-1}\frac{q^{\theta_{n_{a_{k}}}}_{k-2}}{q^{\theta_{n_{a_{k}}}}_{k-1}}\right)A^N_{k, k-1} \le 0 \label{eq:scaled_optim}
\end{align}
Note that $A^\theta_{k, k-1} \ge 0$ for all $\theta\ge\theta_{n_{a_k}}$ and $A^\theta_{k, k-1} \le 0$ for all $\theta<\theta_{n_{a_k}}$. By MLR property of $\{q^\theta_k\}$, 
\begin{align*}
    \theta < \theta_{n_{a_k}} \Rightarrow q^\theta_{k-2} &\ge q^\theta_{k-1}\frac{q^{\theta_{n_{a_k}}}_{k-2}}{q^{\theta_{n_{a_k}}}_{k-1}} \\
    \theta> k-2  \Rightarrow q^\theta_{k-2} &\le q^\theta_{k-1}\frac{q^{\theta_{n_{a_k}}}_{k-2}}{q^{\theta_{n_{a_k}}}_{k-1}}.
    \end{align*}
Thus, in \ref{eq:scaled_optim}, replacing the coefficients $\left(q^\theta_{k-1}\frac{q^{k-2}_{k-2}}{q^{k-2}_{k-1}}\right)$ with $q^\theta_{k-2}$ preserves the inequality, since this weakly decreases the coefficients of positive terms and weakly increases the coefficients of negative terms. Therefore, 
\begin{align}
    \sum_\theta q^\theta_{k-2}A^\theta_{k, k-1} &\le 0\\
    \Rightarrow \sum_{\theta}q^\theta_{k-2}[A^\theta_{k, k-1} + A^\theta_{k-1, k-2}] &= \sum_\theta q^\theta_{k-2}A^\theta_{k-1, k-2} \le 0 
\end{align}
which is optimality of $a_{k-2}$ against $a_k$. Similarly,    
\begin{align}
\left(q^1_{k-1}\frac{q^{\theta_{n_{a_{k-1}}}}_{k}}{q^{\theta_{n_{a_{k-1}}}}_{k-1}}\right)A^1_{k-1, k-2} + \dots + \left(q^N_{k-1}\frac{q^{\theta_{n_{a_{k-1}}}}_{k}}{q^{\theta_{n_{a_{k-1}}}}_{k-1}}\right)A^N_{k, k-1} \ge 0 \label{eq:scaled_optim2}
\end{align}
$A^\theta_{k-1, k-2} \ge 0$ for all $\theta\ge\theta_{n_{a_{k-1}}}$ and $A^\theta_{k-1, k-2} \le 0$ for all $\theta<\theta_{n_{a_{k-1}}}$. Again, by MLR property of $\{q^\theta_k\}_{\theta\in\Theta}$, 
\begin{align*}
    \theta < \theta_{n_{a_{k-1}}} \Rightarrow q^\theta_{k} &\le q^\theta_{k-1}\frac{q^{\theta_{n_{a_{k-1}}}}_{k}}{q^{\theta_{n_{a_{k-1}}}}_{k-1}} \\
    \theta> \theta_{n_{a_{k-1}}} \Rightarrow q^\theta_{k} &\ge q^\theta_{k-1}\frac{q^{\theta_{n_{a_{k-1}}}}_{k}}{q^{\theta_{n_{a_{k-1}}}}_{k-1}}.
\end{align*}
Then, in \ref{eq:scaled_optim2}, replacing the coefficients $\left(q^\theta_{k-1}\frac{q^{k-1}_{k}}{q^{k-1}_{k-1}}\right)$ with $q^\theta_{k}$ preserves the inequality, since this weakly increases the coefficients of positive terms and weakly decreases the coefficients of negative terms. Therefore, 
\begin{align}
    \sum_\theta q^\theta_{k}A^\theta_{k, k-2} &\ge 0\\
    \Rightarrow \sum_{\theta}q^\theta_{k}[A^\theta_{k, k-1} + A^\theta_{k-1, k-2}] &= \sum_\theta q^\theta_{k}A^\theta_{k, k-2} \ge 0. 
\end{align}
which is the optimality of $a_k$ against $a_{k-2}$. Hence $\{A^\theta_{k, k-2}\}$ satisfies optimality. 

Now suppose that $\{A^\theta_{k, l}\}$, where $l < k-2$, satisfies optimality. To complete the proof, it suffices to show that $\{A^\theta_{k, l-1}\} = \{A^\theta_{k, l}\} + \{A^\theta_{l, l-1}\}$ satisfies optimality. Again, $\{A^\theta_{l, l-1}\}$ satisfies optimality by construction. $\{A^\theta_{k, l}\}$ satisfies single crossing, so let $\theta_{n_{k,l}}$ be smallest $\theta$ such that $A^{\theta'}_{k, l}\ge0$ for every $\theta\ge\theta'$. That is, $\theta_{n_{k,l}}$ is the crossing point for $A^{\theta'}_{k, l}$. Consider, 
\begin{align}
\left(q^1_{k}\frac{q^{\theta_{n_{k, l}}}_{k}}{q^{\theta_{n_{k, l}}}_{l}}\right)A^1_{k, l} + \dots + \left(q^N_{k}\frac{q^{\theta_{n_{k, l}}}_{k}}{q^{\theta_{n_{k, l}}}_{l}}\right)A^N_{k, l} \ge 0 
\end{align}

Again, note that, for $k>l$, $\theta< \theta_{n_{k, l}}, q^\theta_{k}\frac{q^{\theta_{n_{k, l}}}_{k}}{q^{\theta_{n_{k, l}}}_{l}} > q^\theta_k$ and for $\theta> \theta_{n_{k, l}}, q^\theta_{k}\frac{q^{\theta_{n_{k, l}}}_{k}}{q^{\theta_{n_{k, l}}}_{l}} < q^\theta_k$. Since $A^\theta_{l, l-1} \ge 0$ for $\theta \ge \theta_{n_{k, l}}$ and $A^\theta_{l, l-1} \le 0$ for $\theta< \theta_{n_{k, l}}$, the inequality is preserved by replacing coefficients with $q^\theta_k$. Therefore, 
\begin{align*}
    \sum_\theta q^\theta_kA^\theta_{k, l-1} =   \sum_\theta q^\theta_k[A^\theta_{k, l}+A^\theta_{l, l-1}] \ge 0
\end{align*}

Finally, consider $\{A^\theta_{k, l}\}$. We have, 
\begin{align*}
    \sum_{\theta}\left(q^\theta_l\frac{q^{\theta_{n_{k, l}}}_{l-1}}{q^{\theta_{n_{k, l}}}_l}\right)A^\theta_{k, l} \le 0.
\end{align*}
By definition, $A^\theta_{k, l} \le 0$ for $\theta<\theta_{n_{k, l}}$ and  $A^\theta_{k, l} \ge 0$ for $\theta\ge\theta_{n_{k, l}}$. Moreover, by MLR, $q^\theta_l\frac{q^{\theta_{n_{k, l}}}_{l-1}}{q^{\theta_{n_{k, l}}}_l}\lesseqgtr
q^\theta_{l-1}$ as $\theta\lesseqgtr\theta_{\theta_{n_{k, l}}}$. Thus, the inequality is once again preserved when the coefficients are replaced with $q^\theta_{l-1}$, and therefore, 
\begin{align*}
    \sum_{\theta}q^\theta_{l-1}A^\theta_{k, l-1} = \sum_{\theta}q^\theta_{l-1}[A^\theta_{k, l-1} + A^\theta_{k, l-1}] \le 0 
\end{align*}
 This ensures the optimality of $\{A^\theta_{k, l-1}\}$, thereby completing the proof. 
\end{proof}

Finally, consider the  system,
 \begin{align*}
    \begin{bmatrix}
     -\bm{1}_1 + \bm{1}_2\\
    -\bm{1}_2 + \bm{1}_3\\
    \vdots\\
    -\bm{1}_{M-1} + \bm{1}_{M}\\
    -\bm{1}_{M+1} + \bm{1}_{M+2}\\
    \vdots\\
    -\bm{1}_{2M-1} + \bm{1}_{2M}\\
    -\bm{1}_{2M+1} + \bm{1}_{2M+2}\\
    \vdots\\
    \vdots\\
    -\bm{1}_{(M-1)(N-1)-1}+ \bm{1}_{(M-1)(N-1)}\\
    -\bm{1}_{(M-1)(N-1)+1}+ \bm{1}_{(M-1)(N-1)+2}\\
    \vdots\\
    -\bm{1}_{(M-1)N-1}+ \bm{1}_{(M-1)N}\\
    \end{bmatrix}\cdot\begin{bmatrix}
        u^1_1 \\u^1_2 \\\vdots \\u^N_M 
    \end{bmatrix}_{MN\times 1}
    = \begin{bmatrix}
        A^1_{2, 1} \\ \vdots \\ A^1_{M, M-1} \\ A^2_{2, 1} \\\vdots\\  A^N_{M, M-1} 
    \end{bmatrix}_{(M-1)N\times 1}
 \end{align*}

There are exactly $(M-1)N$ independent columns in the coefficient matrix on the left hand side.\footnote{Each row on the first matrix on the LHS has exactly two 1's $N^2$ 0's, where the positions of the 1's ensure that the solution to the system induces the corresponding differences $A^\theta_{k, l}$.}  This is also the dimension of the vector on the RHS. So, the rank of the augmented matrix is same as that of the coefficient matrix and thus there exists a utility function $u(a, \theta)$ which induces the differences $\{A^\theta_{k, l}\}$. This completes the construction of a single crossing utility function that rationalizes the data, thereby completing the proof. 

Once again, it follows from the Theorem 1 of \citealp{caplin2015testable} that this utility function rationalizes the data. However, we still need a choice rule and an MLR ordered information structure, that together with the utility function rationalize the data. 

\subsection{Constructing the information structure and the choice rule}

Start by setting the set of signal realizations equal to the set of actions, so $S = A$, and define  
        \begin{align*}
            \gamma^a(\theta) = \frac{\mu_0(\theta)q(a|\theta)}{\sum_{\theta'\in\Theta}\mu_0(\theta')q(a|\theta')}
        \end{align*}
Essentially, $\gamma^a$ is the \textit{posterior that induces action $a$}. Let $\mathcal{A} = \{A(b_1),\dots, A(b_P)\}$ be a partition of $A$ so that actions in the same equivalence class satisfy: $a, a'\in A(b)\Leftrightarrow \gamma^a(\theta) = \gamma^{a'}(\theta), \forall \theta$. Thus, for each member of the partition $A(b)$, there is a unique posterior $\gamma(b)$ that induces one of those actions. To get the information structure, define 
        \begin{align*}
            \pi(\gamma(b)|\theta) = \sum_{a\in A(b)}q(a|\theta)
        \end{align*}
It is then easily seen that, if the data is MLR-ordered, then so is information structure. In particular, note that, if $\theta' > \theta$ and $a' > a$, then, 
\begin{align*}
    \frac{\gamma^{a'}(\theta')}{\gamma^{a'}(\theta)} = \frac{\mu_0(\theta')}{\mu_0(\theta)}\frac{q(a'|\theta')}{q(a'|\theta)} \ge 
    \frac{\gamma^{a}(\theta')}{\gamma^{a}(\theta)} = \frac{\mu_0(\theta')}{\mu_0(\theta)}\frac{q(a|\theta')}{q(a|\theta)} 
\end{align*}
and so, $\mathcal{I} \equiv \{\gamma^a(\theta)\}_{\theta\in\Theta}$ is MLR ordered. 

Now we construct the choice rule. To that end, define, 
\begin{align*}
            C(a|\gamma(b)) = \begin{cases}
                \frac{\sum_{\theta'\in
            \Theta}\mu_0(\theta')q(a|\theta')}{\sum_{a\in A(b)}\sum_{\theta'\in\Theta}\mu_0(\theta')q(a|\theta')}& \text{ if } a\in A(b)\\
                0 & \text{ if } a\notin A(b) 
            \end{cases}
\end{align*}

The procedure for constructing the information structure is exactly as that found in \citealp{caplin2015testable}. A key observation, however, is that when the data is MLR-ordered, then so is the rationalizing information structure. That $\langle u, \mathcal{I}, C\rangle$ rationalize that data - and hence satisfy conditions 2, 3, and 4 of Definition \ref{def:MBEUR} - follow directly from Theorem 1 of \citealp{caplin2015testable} (see Theorem \ref{thm:caplin}). Since we have shown that $u$ is single crossing and $\mathcal{I}$ is MLR ordered, it follows that the tuple  $\langle u, \mathcal{I}, C\rangle$ represents an MBEU-DM that rationalizes the data. 

We have thus established Theorem \ref{thm:Main_theorem}.  

\section{Revealed Informedness} 

Theorem \ref{thm:Main_theorem} suggests a straightforward, but rather interesting application. Recall that an information structure is said to be \textit{Blackwell better} than the other if the former guarantees (weakly) highly ex-ante utility for any bounded utility function \footnote{Ex ante utility refers to expected utility before a signal is realized That is, $\sum_{\theta\in S}\sum_{s\in S}\mu_0(\theta)\mu(s|\theta)u(a*(s), \theta))$. This is often referred to as the value of the information structure.}. Indeed, since Blackwell order requires an information structure to give higher ex-ante value for  \textit{every} utility function, it is, in general, a severely partial order. By restricting attention to smaller classes of utility functions, more complete orders over information structures can be obtained. 

For single crossing utilities, we typically consider the Lehmann order (\citealp{Lehman}), which is weaker than the Blackwell order. In particular, if an information structure $\mathcal{I}_1$ is higher than $\mathcal{I}_2$ in Blackwell order, then $\mathcal{I}_1$ is higher than $\mathcal{I}_2$ in Lehmann order, but the converse may ot hold. In particular, $\mathcal{I}_1$ is higher than $\mathcal{I}_2$ in the Lehmann order if and only if for every single-crossing utility function ex-ante utility is higher under $\mathcal{I}_1$ than under $\mathcal{I}_2$ (see for example, \citealp{Persico}, \citealp{Jewitt}). Below we give a formal definition of Lehmann order.

\begin{definition}
   Let $\mathcal{I}_i \equiv \{\mu^i(\cdot|\theta)\}_{\theta\in\Theta}\subset \Delta(S^i), i =1, 2$ be MLR-ordered information structures. Let $\{F^i_\theta(\cdot)\}_{\theta\in\Theta}, i=1, 2$ be the corresponding distribution functions. Then, $\mathcal{I}_1$ is said to be Lehmann higher, or more accurate, than $\mathcal{I}_2$ if there exists an function $h:S^2\times\Theta\rightarrow S^1$ which is increasing in $\theta$ and, such that, for every $\theta, s^1\in S^1$,\footnote{While this definition is appropriate only for continuous distributions $F^i$, \citealp{Lehman} shows that it is without loss of generality to assume that $F^i$ are continuous. That is, if $F^1, F^2$ are discontinuous, there exist continuous distributions $F^{i^*}$ which are informationally equivalent to $F^i$.}        
   \begin{align*}
       F^1_\theta(h(s^2, \theta)) = F^2_\theta(s^2).  
   \end{align*}
\end{definition} 

Suppose $\mathcal{I}_1 = \{\mu^1(\cdot|s)\}_{s\in S^1}, \mathcal{I}_2\equiv \{\mu^2(\cdot|s)\}_{s\in S^2}$ are MLR ordered, and let $\nu^1, \nu^2$ be the corresponding marginal distributions over signals.\footnote{So, $\nu^i(s) = \int_{\theta\in \Theta}\mu^i(s|\theta)\mu_0(\theta)$ for each $s \in S^i$.} $\mathcal{I}_1$ is Lehmann higher than $\mathcal{I}_2$ if and only if, for every single crossing utility function $u(a, \theta)$,\footnote{See \citealp{Lehman}, theorem 5.1.} 
\begin{align*}
    \int_{s\in S^1}\int_{\theta\in \Theta}u(a^*(s), \theta)d\mu^1(\theta|s)d\nu^1(s)\ge   \int_{s\in S^2}\int_{\theta\in \Theta}u(a^*(s), \theta)d\mu^2(\theta|s)d\nu^2(s). 
\end{align*}

Our main result gives the conditions under which rationalizing posteriors and single crossing utility functions may be constructed from the observed state conditional stochastic choice data. Suppose that two distinct datasets are MBEU-rationalizable. If the analyst can somehow conclude that one rationalizing information structure, say $\mathcal{I}_1$ is Lehmann higher than the other, $\mathcal{I}_2$, then, the analyst can infer that any MBEU-DM rationalizing the former has better information than the one rationalizing the latter.   

\begin{definition}
    Fix some class  $
    \mathcal{U}$ of utility functions, and consider two datasets $\mathcal{D}^1 = \{q^1(\cdot|\theta)\}$, $\mathcal{D}^2 =\{q^2(\cdot|\theta)\}$. We say that the DM corresponding to $\mathcal{D}^1$ is $\mathcal{U}-$\textbf{revealed more informed} than the DM corresponding to $\mathcal{D}^2$ if, there exist information structures  $\mathcal{I}_1, \mathcal{I}_2$ such that, for any utility functions $u_1, u_2\in\mathcal{U}$, if $(\mathcal{I}_1, u_1), (\mathcal{I}_2, u_2)$ are BEU-rationalizations of $\mathcal{D}^1, \mathcal{D}^2$, respectively, then, 
    \begin{align*}
    \mathbb{E}_{\mathcal{I}_1}u_1 \ge \mathbb{E}_{\mathcal{I}_2}u_1\\  \mathbb{E}_{\mathcal{I}_1}u_2 \ge \mathbb{E}_{\mathcal{I}_2}u_2
    \end{align*}     
\end{definition}

Clearly, if $\mathcal{U}$ is the family of all single crossing utility functions,  the MBEU-DM corresponding to $\{q^1(\cdot|\theta)\}$ is \textbf{revealed more informed} than the MBEU-DM corresponding to $\{q^2(\cdot|\theta)\}$ if and only if $\mathcal{I}_1$ is Lehmann higher than $\mathcal{I}_2$. Then, we have the following result:

\begin{theorem}
     Let $\mathcal{D}^1 = \{q^1(\cdot|\theta)\}, \mathcal{D}^2 = \{q^2(\cdot|\theta)\}$ be MLR ordered and such that $\mathcal{D}^1$ is Lehmann higher than $\mathcal{D}^2$. Then, the MBEU-DM corresponding to $\mathcal{D}^1$ is revealed more informed than the MBEU-DM corresponding to $\mathcal{D}^2$.      
\end{theorem}

\begin{proof}
    See appendix. 
\end{proof}

While the proof is entirely trivial, we contend that this result is nonetheless interesting. To underscore the argument, it is instructive to compare our rationalizability result with that of \cite{caplin2015testable}. Indeed, if one is interested in BEU-rationalizability alone, and the analyst has datasets $\mathcal{D}^1, \mathcal{D}^2$ such that $(i)$ both $\mathcal{D}^1, \mathcal{D}^2$ are BEU-rationalizable and $(ii)$ $\mathcal{D}^1$ Blackwell dominates $\mathcal{D}^2$, then the analyst can conclude that the BEU-DM corresponding to $\mathcal{D}^1$ is more informed than the BEU-DM corresponding to $\mathcal{D}^2$. But, as pointed out above, it is not known whether, in general, a solution to the NIAS inequalities exists. Secondly, Blackwell order is quite incomplete. On the other hand, the definition of Lehmann order \textit{requires} the datasets to be MLR ordered - which guarantees the existence of an MBEU-DM.   

\section{Discussion}

\paragraph{Independent tests of the hypotheses} This paper restricts attention to the test of joint hypothesis of an MLR ordered information structure and an SC utility function. However, the proof of sufficiency does suggest a test of the hypothesis of single crossing alone. 

Consider any two actions $a_k, a_l$ and for any $\theta$, denote $A^\theta_{k, l} = u^\theta_{k} - u^\theta_{l}$. As noted above, the proof of sufficiency relies on the main result of \citealp{caplin2015testable}, who show that rationalizability with a Bayesian expected utility DM is equivalent to the existence of a solution to a certain system of inequalities (see \ref{sec:appendixa}). We show that MLR property ensures that a solution exists, which is \textit{also} single crossing. 

Suppose, however, that one was interested is testing SCP alone. Then, for any pair of actions $k > l$, BEU rationalizability requires 
\begin{align*}
    \sum_\theta q_k^\theta u^\theta_{k} \ge \sum_\theta q_k^\theta u^\theta_{l}\\
    \sum_\theta q_l^\theta u^\theta_{l} \ge \sum_\theta q_l^\theta u^\theta_{k} 
\end{align*}
where the inequalities are strict for at least one pair of actions. Without loss of generality, suppose one of these is strict. These can then be rewritten as,
\begin{align*}
    \sum_\theta q_k^\theta[ u^\theta_{k} - u^\theta_{l}] \ge 0 \ge 
    \sum_\theta q_l^\theta[ u^\theta_{k} - u^\theta_{l}] 
\end{align*}
where one of two inequalities is strict. Suppose that $u_k^{\theta_N} - u_l^{\theta_N} < 0$. Then, SCP requires that $u_k^{\theta} - u_l^{\theta} < 0$ for all $\theta$ - however, this would imply that numbers $u^\theta_k, u^\theta_l$ cannot be a solution to the inequality $\sum_\theta q^\theta_k[u^\theta_k-u^\theta_l]\ge 0$. The analyst must then reject the single crossing property. Similarly, if $u_k^{\theta_1} - u_l^{\theta_1} > 0$, then $u_k^{\theta} - u_l^{\theta} > 0$ for every $\theta$ and therefore the numbers $u^\theta_k, u^\theta_l$ cannot be a solution to the inequality $\sum_\theta q^\theta_l[u^\theta_k-u^\theta_l]\le 0$.                
However, under uncertainty, almost all of the bite of single crossing property rests upon the information structure being MLR ordered. Indeed, one implication of the result of \citealp{athey2002monotone} is that, if the information structure is not MLR ordered, then there exists a utility function satisfying the SCP such that MCS does not hold. Thus, under uncertainty, most of the bite of SCP comes when combined with an MLR ordered information structures.      

Furthermore, without uncertainty, the only observable implication is monotonicity of chosen actions with respect to the exogenous covariate (e.g. known state), as shown in \citealp{lazzati2018nonparametric}. As pointed out above, the single crossing property is an ordinal generalization of the property of increasing differences. The latter is intimately related to supermodularity.\footnote{A function  $f:X\times\Theta\rightarrow\mathbb{R}$ is said to be supermodular if $f(x, \theta) + f(x', \theta') \le f((x, \theta)\vee (x', \theta'))+f((x, \theta)\wedge (x', \theta'))$, where $X\times\Theta$ is ordered according to the product ordering induced by the orderings over the sets $X$ and $\Theta$. If the functions $x\mapsto f(x, \theta)$ and $\theta\mapsto f(x, \theta)$ are both supermodular for any $\theta$ and $x$ respectively, and if $f$ has increasing differences then $f$ is supermodular. See section 2.6.1 of \citealp{topkis1998supermodularity}.} Interestingly, \citealp{chambers2009supermodularity} argue that quasi-supermodularity, the ordinal generalization of the notion of supermodularity is indistinguishable from the latter if the underlying domain is finite and the function is assumed to be monotonic. In several settings, therefore, this property is not refutable. A novelty of the results reported here is thus establishing an observational non-equivalence between monotonicity and single crossing property for certain kinds of data.

\paragraph{Identifying testable restrictions} The proof exhibits that as long as the data is MLR-ordered, one can always find a solution to the system in \ref{eq:NIAS}. Therefore, the MLR condition is the \textit{only} condition required for rationalizing the data by an SCP utility function and MLR information structure. The method of proof here suggests a possible way of approaching problems involving inference in models of information. Existence of a Bayesian DM is equivalent to the existence of a solution to a system of inequalities. Then, suppose an analyst wants to test additional restrictions on this model. The restrictions could either be on the information structure, or the utility function, or both. Whether such restrictions have any implications on stochastic choice data in addition to the NIAS inequalities, is determined by whether these restrictions imply the existence of a solution to the NIAS inequalities, or if they ``shrink" the set of solutions. The approach in this paper therefore suggests a method to assess whether a particular restriction is testable.

\citealp{rambachan2024identifying} shows how to combine the NIAS inequalities with known methods of partial identification to obtain a novel theory of identifying prediction mistakes with stochastic choice data. Our results can thus be possibly used to further sharpen the identification bounds by imposing some shape restrictions.

%\paragraph{Information acquisition} 
%This paper is also silent on how exactly the information is acquired. Following the works of \citealp{matvejka2015rational}, \citealp{caplin2013behavioral}, There is now a large literature on costly information acquisition. While in general these problems are difficult to solve, some advances have been achieved in the comparative statics of such problems, for instance in \citealp{curello2022comparative} and \citealp{whitmeyer}. \citealp{caplin2015revealed} is an extension of \citealp{caplin2015testable} to the case of costly information acquisiton.%

\clearpage

\onehalfspacing

\clearpage

\section*{Appendix A. BEU-Rationalizability} \label{sec:appendixa}
\addcontentsline{toc}{section}{Appendix A}

This section briefly presents the result of \citealp{caplin2015testable} since the latter is the result on which our paper builds. A BEU-rationalization of a dataset is $\mathcal{D}$ is a tuple $\langle u, \pi, C\rangle$ such that conditions $2-4$ in Definition \ref{def:MBEUR} hold. 

\begin{theorem}[\citealp{caplin2015testable}, Theorem 1]
    Let $\mathcal{D} \equiv \langle A, \Theta, \mu_0, \bm{q}\rangle$. Then there exists a BEU-rationalization $\langle u, \pi, C\rangle$ of $\mathcal{D}$ if and only if there exists a function $U:A\times\Theta\rightarrow\Theta$ such that, 
    \begin{align}
    \sum_{\theta\in\Theta}\mu_0(\theta)q(a|\theta)U(a, \theta) \ge \sum_{\theta\in\Theta}\mu_0(\theta)q(a|\theta)U(b, \theta) \label{eq:NIAS}
    \end{align}
    for every $a, b\in A$, where the inequality holds strictly for at least one pair of actions. \label{thm:caplin}
\end{theorem}
 Since $A$ and $\Theta$ are finite, the system of inequalities \ref{eq:NIAS} is a finite system of inequalities. Thus, if there are no solutions to this system, the Bayesian expected utility hypothesis is falsified by observed data. The next section shows that, in binary state case, a single condition on $q\{(\cdot|\theta)\}_{\theta\in\Theta}$ is equivalent to MBEU-rationalization.       

\section*{Appendix B. Proof for the Binary Case} \label{sec:appendixb}
\addcontentsline{toc}{section}{Appendix B}
\begin{proof} 
    Suppose that the data is MLR ordered. Then, $\frac{q(a_2|\theta_2)}{q(a_1|\theta_2)} \equiv \frac{q^2_2}{q^1_2} > \frac{q^2_1}{q^1_1}\equiv \frac{q(a_2|\theta_1)}{q(a_1|\theta_1)}$. Or equivalently, $\frac{q^1_1}{q^1_2} > \frac{1-q^1_1}{1-q^1_2}$. One can then find numbers $u_j^i$ such that
    \begin{align}
            \frac{\mu q^1_1}{(1-\mu)q^1_2} \ge \frac{u^2_2 - u^1_2}{u^1_1 - u^2_1}\ge \frac{\mu(1-q^1_1)}{(1-\mu)(1-q^1_2)} \label{eq:bin_1}
        \end{align}   
    where at least one of the two inequalities holds strictly. Then,   
    \begin{align}
        \mu q^1_1u^1_1 + (1-\mu)q^1_2u^1_2 &\ge         \mu q^1_1u^2_1 + (1-\mu)q^1_2u^2_2 \label{eq:bin_nias1}\\
        \mu(1-q^1_1)u^2_1 + (1-\mu)(1-q^1_2)u^2_2 &\ge         \mu(1-q^1_1)u^1_1 + (1-\mu)(1-q^1_2)u^1_2\label{eq:bin_nias2}
    \end{align}
    where at least one of the two holds strictly. Since there are only two actions and two states, one can easily construct a consistent information structure. Indeed, consider two signal labels $s_1$ and $s_2$ and define
    \begin{align*}
        \gamma_1 \equiv (\gamma(\theta_1|s_1), \gamma(\theta_2|s_1)) &= \left(\frac{\mu q^1_1}{\mu q^1_1 + (1-\mu)q^1_2}, \frac{(1-\mu) q^1_2}{\mu q^1_1 + (1-\mu)q^1_2}\right)\\
        \gamma_2 \equiv (\gamma(\theta_1|s_2), \gamma(\theta_2|s_2)) &= \left(\frac{\mu(1-q^1_1)}{\mu(1-q^1_1) + (1-\mu)(1-q^1_2)}, \frac{(1-\mu)(1-q^1_2)}{\mu(1-q^1_1) + (1-\mu)(1-q^1_2)}\right) 
    \end{align*} 
with $\pi(\gamma_1) = \mu q^1_1 + (1-\mu)q^1_2$ and $\pi(\gamma_2) = \mu(1-q^1_1) + (1-\mu)(1-q^1_2)$. Finally, define the choice rule $C(\gamma_1) = \delta_{a_1}$ and $C(\gamma_2) = \delta_{a_2}$.\footnote{$\delta_x$ denotes point mass probability or the \textit{Dirac measure} at $x$.}

It can be easily checked that $\langle u, \pi,  C\rangle$ is an MBEU-rationalization. Bayes-plausibility follows from construction. Optimality follows from equations \ref{eq:bin_nias1}, \ref{eq:bin_nias2} and the definition of $(\gamma_1, \gamma_2)$. Consistency follows directly by construction of $\pi$. Finally, choosing $(u_j^i)_{i, j\in \{1, 2\}}$ such that ${u^1_1 - u^2_1} > 0$ ensures that SCP holds.                  

Conversely, suppose $\langle u, \pi, C\rangle$ is an SC utility function which rationalizes the data. By optimality, for any $\gamma\in\text{supp}(\pi)$ and any $a, b\in A$,  
\begin{align*}
    C(a|\gamma)[\gamma(\theta_1)u(a, \theta_1) + \gamma(\theta_2)u(a, \theta_2)] &\ge C(a|\gamma)[\gamma(\theta_1)u(b, \theta_1) + \gamma(\theta_2)u(b, \theta_2)].
\end{align*}
Aggregating across posteriors, 
\begin{align*}
    \sum_{\gamma\in\text{supp}(\pi)}C(a|\gamma)[\gamma(\theta_1)u(a, \theta_1) + \gamma(\theta_2)u(a, \theta_2)] &\ge \sum_{\gamma\in\text{supp}(\pi)}C(a|\gamma)[\gamma(\theta_1)u(b, \theta_1) + \gamma(\theta_2)u(b, \theta_2)]\\
    \sum_{i=1, 2}u(a, \theta_i)\sum_{\gamma\in\text{supp}(\pi)}C(a|\gamma)\gamma(\theta_i) &\ge \sum_{i=1, 2}u(b, \theta_i)\sum_{\gamma\in\text{supp}(\pi)}C(a|\gamma)\gamma(\theta_i).  
\end{align*}
By consistency, the last inequality reduces to, 
\begin{align*}
    \sum_{i=1, 2}q(a|\theta_i)u(a, \theta_i) &\ge \sum_{i=1, 2}q(a, \theta_i)u(b, \theta_i).
\end{align*}
Since this holds for any pair of actions,  
\begin{align*}
    q(a_1|\theta_1)u(a_1, \theta_1) + q(a_1|\theta_2)u(a_1, \theta_2) \ge     q(a_1|\theta_1)u(a_2, \theta_1) + q(a_1|\theta_2)u(a_2, \theta_2)\\
    q(a_2|\theta_1)u(a_2, \theta_1) + q(a_2|\theta_2)u(a_2, \theta_2) \ge     q(a_2|\theta_1)u(a_1, \theta_1) + q(a_2|\theta_2)u(a_1, \theta_2). 
\end{align*}
where at least one of the two holds strictly. These can be rewritten as, 
\begin{align*}
    q(a_1|\theta_1)[u(a_1, \theta_1) - u(a_2, \theta_1)] + q(a_1|\theta_2)[u(a_1, \theta_2) - u(a_2, \theta_2)] \ge 0\\
    q(a_2|\theta_1)[u(a_2, \theta_1)-u(a_1, \theta_1)] + q(a_2|\theta_2)[u(a_2, \theta_2) - u(a_1, \theta_2)] \ge 0.  
\end{align*}

Now, since $u$ is SCP, note that $u(a_1, \theta_1) - u(a_2, \theta_1)$ must be positive. Otherwise, $u(a_1, \theta_2) - u(a_2, \theta_2)$ must also be negative violating the first inequality above. Thus, the above  inequalities can be rewritten as, 
\begin{align}
            \frac{q(a_1|\theta_1)}{q(a_1|\theta_2)} \ge \frac{u^2_2 - u^1_2}{u^1_1 - u^2_1}\ge \frac{q(a_2|\theta_1)}{q(a_2|\theta_2)} 
\end{align}
where at least one of the two inequalities holds strictly. But this implies,   $\frac{q(a_1|\theta_1)}{q(a_1|\theta_2)} > \frac{q(a_2|\theta_1)}{q(a_2|\theta_2)}$ and thus the data is MLR-ordered.  \\
\end{proof}

\section*{Appendix C. Proof of Theorem 4} \label{sec:appendixc}
\addcontentsline{toc}{section}{Appendix C} 

\begin{proof}
     Consider the datasets $\mathcal{D}^i\{q^i(\cdot|\theta)\}_{\theta\in\Theta}$, $i=1, 2$, and let $F^i_\theta$ be the cumulative distribution corresponding to $q^i(\cdot|\theta)$. From Lehmann, we know that it is without loss of generality to assume that $F^i_\theta$ is continuous for each $\theta$. So, $F^i_\theta(a) = \int_{a'<a}q^i(a'|\theta)da'$. Recall that, the constructed posterior rationalizing the dataset was defined as $\gamma^a_i(\theta) = \frac{q^i(a|\theta)\mu_0(\theta)}{\sum_{\theta'\in\Theta}q^i(a|\theta')\mu_0(\theta)} = \frac{q^i(a|\theta)\mu_0(\theta)}{\nu^i(a)}$ - where $\nu^i$ is the marginal distribution over actions corresponding to $\mathcal{D}^i$. But then, the probability of obtaining a signal $a$ when the state is $\theta$ is simply $\gamma_i(\theta|a) = \frac{\gamma^a_i(\theta)\nu^i(a)}{\mu_0(\theta)} = q^i(a|\theta)$. Therefore, the cumulative distribution of $\gamma_i(a|\theta)$ - the state conditional distributions over signal realizations - is simply $F^i_\theta$, and thus the rationalizing information structures are Lehmann ordered. That is, the MBEU-DM corresponding to $\mathcal{D}^1$ is revealed more informed than the MBEU-DM corresponding to $\mathcal{D}^2$.  
\end{proof}

\setlength\bibsep{0pt}
\bibliographystyle{aea}
\bibliography{bib2}

\end{document}